%% file: paper.tex
\documentclass[11pt,twocolumn]{article}
\usepackage{fullpage}
\usepackage{times}
\usepackage[cmex10]{amsmath}
\usepackage{amsfonts,amssymb}
\usepackage{comment}
\usepackage[dvips]{graphicx}
\usepackage[usenames,dvips]{color}
\usepackage[colorlinks=true,citecolor=blue,urlcolor=BrickRed,
linkcolor=magenta]{hyperref}
\usepackage{url}
\usepackage[small,it]{caption}
\usepackage[font=footnotesize]{subfig}

\usepackage{amsthm}
\usepackage{fixltx2e}
\usepackage{enumerate}
\usepackage{paralist}
\usepackage{breakurl}
\usepackage{cite}
\usepackage{multirow}
\usepackage{xspace}
\usepackage{authblk}

\usepackage[ruled,lined,linesnumbered]{algorithm2e}

\newtheorem{theorem}{Theorem}[section]

\newtheorem{lemma}[theorem]{Lemma}

\newcommand{\ie}{{\em i.e.}\xspace}

\newcommand{\whp}{{\em w.h.p.}\xspace}

\hypersetup{
 unicode=false,        
 pdfnewwindow=true,    
 colorlinks=true,	  
 citecolor=blue,
 urlcolor=BrickRed,
 linkcolor=magenta,
 linkcolor=BrickRed,   
 pdfmenubar=true     
}

\def\algr{\textsc{A}\xspace}
\def\mstsinr{\textsc{MST-SINR}\xspace}

\begin{document}
\title{\bf A Fast Distributed Approximation Algorithm for Minimum Spanning Trees in the SINR Model}

\author[1]{Maleq Khan}
\author[1]{V.S. Anil Kumar}
\author[2]{Gopal Pandurangan}
\author[1]{Guanhong Pei}
\affil[1]{Virginia Bioinformatics Institute, Virginia Polytechnic Institute and State University, Blacksburg, VA, USA}
\affil[2]{Division of Mathematical Sciences, Nanyang Technological University, Singapore and Department of Computer Science, Brown University, Providence, RI, USA}
\date{}

\maketitle
\thispagestyle{empty}

\begin{abstract}

A fundamental problem in wireless networks is the
\emph{minimum spanning tree} (MST) problem:
given a set $V$ of wireless nodes, compute a spanning tree $T$,
so that the total cost of $T$ is minimized.
In recent years, there has been a lot of
interest in the physical interference model based on SINR constraints.
Distributed algorithms are especially challenging in the SINR model, because
of the non-locality of the model.

In this paper, we develop a fast  distributed  approximation algorithm for MST
construction in an SINR based distributed computing model. For an $n$-node network,
our algorithm's running time is $O(D\log{n}+\mu\log{n})$ and produces a spanning
tree whose cost is within $O(\log n)$ times the optimal (MST cost),
where $D$ denotes the diameter of the disk graph obtained by using the
maximum possible transmission range, and $\mu=\log{\frac{d_{max}}{d_{min}}}$ denotes the ``distance diversity''
w.r.t. the largest and smallest distances between two nodes.
(When $\frac{d_{max}}{d_{min}}$ is $n$-polynomial, $\mu = O(\log n)$.)
Our algorithm's running time is essentially optimal (upto a logarithmic factor), since computing {\em any}
spanning tree takes $\Omega(D)$ time; thus our algorithm produces a low cost spanning tree in time only a logarithmic factor more than the  time to compute a spanning tree.
The distributed scheduling complexity of the spanning tree resulted from our algorithm is $O(\mu \log n)$.
Our algorithmic design techniques can be useful in designing efficient distributed algorithms for related  ``global" problems in wireless networks in the SINR model.
\end{abstract}


\input{introduction}

\input{model}

\input{bounded-power}

\input{conclusion}

\bibliographystyle{plain}
\bibliography{ref,energy}



\end{document}

%% file: introduction.tex
\section{Introduction}
\label{sec:intro}

Emerging networking technologies such as ad hoc wireless and sensor networks  operate under inherent  resource constraints such as  power,  bandwidth etc. A distributed algorithm which
exchanges a large number of messages and takes a lot of time can consume a relatively
large amount of resources, and is not very suitable in a resource-constrained
network.
Also, the topology of these networks can  change dynamically.  Communication cost and
running time is especially crucial in  a  dynamic setting.
 Hence it becomes necessary  to design efficient distributed algorithms for various network optimization problems  that have
low communication  and time complexity, even possibly at the cost of a reduced  quality of solution.
(For example, there is not much point in having an optimal algorithm if  it takes too much time, since the topology could have changed by that time.)
For this reason, in the distributed context,  such algorithms are motivated even for network optimization problems that are not NP-hard, e.g., minimum spanning tree (MST), shortest paths (see e.g.,   \cite{elkin-survey}). However, much of the theory of distributed approximation
for various fundamental problems such as MST have been developed in the context of wired networks,
and not for wireless networks under a realistic interference model.

In this paper we focus on one of the most fundamental distributed computing problems in wireless networks, namely the
\emph{Minimum Spanning Tree (MST)} problem. It is a recurring sub-problem in many
network and protocol design problems, and there has been a lot of work on
 distributed algorithms for computing the MST. Computing an MST by a distributed algorithm is a fundamental task,
as the following distributed computation can be carried over the
best backbone of the communication graph.
Two important applications of an MST in wireless networks are
broadcasting and data aggregation. An MST can be used as broadcast
tree to minimize energy consumption since it minimizes
$\sum_{(u,v)\in T}{d^\alpha(u,v)}$. It was shown
in~\cite{Ambuhl-opt-energy,Wan-min-energy,clementi-min} that
broadcasting based on MST consumes energy within a constant factor
of the optimum.

Much of the traditional work on distributed algorithms \cite{lynch,peleg} has focused on
the CONGEST model of message passing, in which
a node can communicate with all its neighbors in one time step. This model is more suited for
wired networks, and the complexity of many fundamental problems (e.g., MST \cite{DistMst:Gallager}, shortest paths \cite{dist-bellman-ford}, etc.) is very well understood in the wired model.
This model does not capture \emph{interference}, an inherent aspect of wireless networks,
which causes collisions when ``close-by'' nodes transmit.
The Radio Broadcast Network (RBN) Model (see, e.g., \cite{alon:jcss91, gpx:podc05}) was developed
specifically to address such issues.  This model is defined on a unit disk communication
graph $G=(V,E)$, and the transmission from a node $u$ to its neighbor $v$ is successful,
provided no other neighbor $w\neq u, w\in N(v)$ transmits at the same time. Therefore, nodes
that can transmit simultaneously form some kind of independent sets, and this model has been
studied extensively over the last two decades. It is known that the RBN model is significantly
different from the CONGEST model \cite{peleg, gpx:podc05}, and its ``local'' structure has been used
in designing efficient algorithms for many fundamental graph problems.

Though the RBN model is much closer model of wireless transmission than the CONGEST model, it
still does not capture several crucial features \cite{gow:mobihoc07, moscibroda:infocom06}:
even if a node $v$ receives packets from two neighbors $u$ and
$w$, the transmissions are not necessarily lost --- this depends on the ratio of the
strength of the $u$'s signal to all other noise being above a threshold. On the other hand,
simultaneous transmissions by a number of ``well-separated'' nodes (which is possible in the
RBN model) might be infeasible in reality if the above-mentioned ratio of signal strengths is
violated at some receiver. Therefore, the complexity of wireless transmission is not captured
by the simple ``local'' constraints of the RBN model.  A different model, called the
\emph{Physical Interference model based on SINR (signal-to-interference \& noise ratio) constraints} (henceforth, referred to as the
SINR model) has been proposed to rectify the weaknesses of the RBN model
\cite{gow:mobihoc07, moscibroda:infocom06, moscibroda:mobihoc06}. The SINR model more accurately captures the
physical nature of wireless networks and has been used in a number of
recent studies (e.g., see the recent survey of Lotker and Peleg \cite{sinr-survey} and the references therein).
However, unlike traditional (wired and RBN) models, distributed algorithms are especially challenging to design and analyze in the SINR model, because of the non-local nature of the model.  In particular, to the best of our knowledge no prior distributed algorithms have been designed  for ``global" problems
in the SINR model. Algorithms in the RBN model do not translate to efficient
algorithms in the SINR model, whose non-locality makes it much harder.

Centralized algorithms for various fundamental
problems, such as independent sets, coloring, dominating set, spanning
tree have been developed in this model in recent years, e.g.,
\cite{Goussevskaia+:INFOCOM09, Wan+:WASA09, Halldorsson+:ICALP09}.  However,
distributed algorithms which satisfy the SINR constraints at each step
are only known for very few problems (e.g., \cite{Fanghanel+:ICALP09,Kesselheim+:DISC10, Halldorsson+:ICALP11}). Furthermore,  the distributed algorithms known are for ``local" problems such as independent set and coloring and not for ``global" problems.
Global
problem are those that require an algorithm to ``traverse" the entire
network. Classical ``global" problems include spanning tree,  minimum spanning
tree, shortest path etc. Network diameter is an inherent lower bound
for such problems.  We note that the known algorithms in the SINR model for
spanning tree problem and the connectivity problem
\cite{Moscibroda+:INFOCOM06,Moscibroda+:MobiHoc06, avin09} are in the  centralized setting.

In this paper, we describe the first distributed  algorithm in the SINR
model for approximate MST construction. The SINR constraints are satisfied
at all steps of the algorithm, which gives a spanning tree with cost
$O(\log{n})$, relative to the MST cost.  The running time is within
a polylogarithmic time of the lower bounds in the RBN model. We discuss
our results in more detail below.
\begin{asparaenum} [(1)]
\item
The running time of our algorithm is $O(D\log{n}+\mu\log{n})$, with high probability,
where $D$ is the diameter of the graph restricted by the maximum  range. This
is optimal up to a polylogarithmic factor, since computing {\em any} spanning tree
takes $\Omega(D)$ time. In particular, if the maximum power level is
unconstrained, the running time is $O(\mu\log{n})$.
The spanning tree produced by our algorithm
has a low distributed scheduling complexity (defined in Section~\ref{sec:sched-complexity}) of $O(\mu \log n)$, \ie, the transmission requests on the edges can be scheduled distributedly in $O(\mu \log n)$ time steps, for any orientation of the edges.
\item
Our main technical contribution is the adaptation of the
technique of ``Nearest Neighbor Trees'' (NNT) \cite{khan_tpds, khan_tcs, khan_distcomp} to the SINR model.
This technique results in spatial separation at each step, which
helps in ensuring the SINR constraints. Our algorithmic design technique can be useful in designing efficient distributed algorithms
for related  ``global" problems in wireless networks in the SINR model.
\end{asparaenum}

All prior distributed algorithms in the SINR model have been restricted to
scheduling kind of problems (e.g., independent sets, coloring and broadcasting).
Our result is the first distributed algorithm for a problem not in the above
class, and might be useful in other network design problems.

%% file: model.tex
\section{Preliminaries and distributed primitives}
\label{sec:model}
\begin{table} [htbp]
\centering
\vspace{-0.05in}
\caption{Notation.}\label{tab:notation}
\fontsize{9.5}{12pt}\selectfont
\begin{tabular} {||c|c||c|c||} \hline
$V$ & set of nodes & $n$ & \#nodes  \\ \hline
$D$ & network diameter & $d(u,v)$ & dist. between $u,v$ \\ \hline
$\alpha$ & path-loss exponent & $\beta$ & SINR threshold \\ \hline
$N$ & background noise & $P$ & transmission power \\ \hline
$\mu$ & distance diversity & $r$ & transmission range \\ \hline
\end{tabular}
\vspace{-0.1in}
\end{table}
\normalsize

Let $V$ denote a set of tranceivers (henceforth, referred to as nodes) in the Euclidean plane, and  $d(u,v)$ the Euclidean distance between nodes $u,v$. Let $d_{min}$ and $d_{max}$ denote the smallest and the largest distances between any two nodes respectively. $d_{min}$ is a constant as a result of the dimension of a node as a wireless device. We normalize the distances, such that $d_{min} = 1$.
We use $D$ to denote the diameter of the disk graph obtained based on $V$ by using the
maximum possible transmission range. We say a node $u$ is in the ball $B(v, r)$ of node $v$ if and only if $d(u, v) \leq r$.
Let $P(u)$ denote the transmission power chosen by node $u$, with
a maximum power level possible, denoted by $P_{max}$. We also assume that
the nodes are capable of ``adaptive power control'', which means that they
can transmit at any power level in the range $[0, P_{max}]$.  We assume
the commonly used path loss models \cite{Goussevskaia+:MobiHoc07, chafekar:mobihoc07}, in which the transmission from $u$ to $v$ is possible only if:
	$\frac{P(u)}{d^{\alpha}(u,v)} \big/ N \geq \beta$,
where $\alpha>2$ is the ``path-loss exponent'', $\beta>1$ is the minimum SINR
required for successful reception, and $N$ is the background noise
(note that $\alpha, \beta$ and $N$ are all constants).
This inequality also defines under a certain transmission power $P$ the \emph{transmission range},
which is the threshold distance beyond which two nodes cannot communicate with each other, and which equals $\sqrt[\alpha]{P/(N \beta)}$.
To reduce notational clutter, we will say that the transmission range $r$
associated with a transmission power $P$ is $r=(P/c)^{1/\alpha}$, for
a constant $c$. Let $r_{max}=(P_{max}/c)^{1/\alpha}$ denote the maximum
transmission range of any node at the maximum power level.
W.l.o.g., we assume $r_{max} \leq d_{max}$ in our algorithms.

We say that a set $S$ of nodes in the plane form a ``constant density set''
w.r.t. range $r$ if there are $O(1)$ nodes within the ball $B(v,r)$ of any node $v$,
where $v$ is the center and $r$ is the radius.
A set $S'\subset S$ is said to be a \emph{constant density dominating set} for $S$
w.r.t. range $r$, if $S'$ is a constant density set, and for each $v\in S$,
there exists $u\in S'$ such that $v\in B(u, r)$. Given a set $S$ of nodes
and range $r$, we define $G_r(S)= \big(S, E=\{(u,v): u, v\in S, d(u,v)\leq r\}\big)$ to be
the graph induced by $S$ with range $r$.

\noindent
\textbf{Wireless interference}.
We use physical interference model based on geometric SINR constraints (henceforth
referred to as the SINR model), where
a set $L$ of links can make successful transmissions simultaneously if and
only if the following condition holds for each $l=(u,v) \in L$:
\begin{equation}
\label{eqn:sinr}
	\frac{\frac{P(u)}{d^{\alpha}(u,v)}} {\sum_{u' \in V' \setminus \{u\}} \frac{P(u')}{d^{\alpha}(u',v)} + N} \geq \beta,
\end{equation}
where $V'$ is the set of transmitting nodes. Such a set $L$ is said to be \emph{feasible} in the context.

\noindent
\textbf{\mstsinr: The minimum spanning tree problem under the SINR model}.
Given a set $V$ of wireless nodes, the goal is to find a spanning tree $T$,
such that the total cost of $T$ is minimized,
w.r.t. a cost function $cost(u,v) = d(u,v)$ for any pair of nodes $(u,v)$;
for a set $E$ of edges, we define $cost(T) = \sum_{(u,v) \in E} cost(u,v)$.
We focus on developing distributed approximation algorithms. We say an
algorithm gives a $\gamma$-approximation factor if it constructs a spanning tree $T$,
with $cost(T) \leq \gamma \cdot cost(MST)$, where $MST$ represents an optimum solution.
In the problem we study here, we only require the tree to be constructed
implicitly --- we assume we have a sink node $s$, and we
define a \emph{parent}, $par(v)$ for all nodes $v\neq s$,
such that the set of edges $\{(v, par(v)): v\in V, v\neq s\}$
form a spanning tree. Each node only needs to know the identity of its parent.
The goal of this paper is to design an algorithm for computing an approximate
MST in the distributed computing model based on SINR constraints (described
below); we do not require the transmissions (in either direction) on all the edges in the tree to be simultaneously
feasible in the SINR model.

\noindent
\textbf{Distributed computing model under the SINR model}.
Traditionally, distributed algorithms for wireless networks have been studied in the
radio broadcast model \cite{d2model, Moscibroda+:PODC05, Schneider+:PODC08} and its
variants. The SINR based computing model is relatively recent, and has not been
studied that extensively. Therefore, we summarize the main aspects and assumptions
underlying this model below.
\begin{asparaenum} [(1)]
\item
The network is synchronized and for simplicity we assume all time slots have the
same length.
\item
The graph $G_{r_{max}/c}(V)$ induced by $V$ with range $r_{max}/c$
is connected, i.e., the graph can be connected by using a transmission power level
of $P_{max}/c'$ which is slightly less than the maximum but within a constant
factor.
\item All nodes have a common estimate of $n$ within
a $n^c$ for some constant $c$.
\item
All nodes share a common estimate of $d_{min}=1$ and $d_{max}$, the minimum and maximum distances between nodes.
We use $\mu=\log{\frac{d_{max}}{d_{min}}}$ to denote the ``distance diversity'' (similar to the ``link diversity'' in \cite{Goussevskaia+:MobiHoc07}),
which is the number of classes of similar length edges
into which the set of all edges can be partitioned. It is common to assume that
$\mu=O(\log{n})$.
\item
As mentioned earlier, we assume nodes are equipped to be able to do adaptive power control, i.e.,
each node $v$ can transmit at any power level $P\in[0, P_{max}]$.
\end{asparaenum}

\subsection{Distributed primitives}
\label{sec:primitives}

We discuss two primitives that are needed in our MST algorithms for
local broadcasting and dominating sets. The results of
\cite{Goussevskaia+:DialMPOMC08, scheideler+:MobiHoc08} directly provide
us efficient implementations for these two problems.

\noindent
1.
\textbf{Local broadcasting}:
The \emph{local broadcasting range} $r_b$ is the distance up to which nodes intend to broadcast
their messages.
We say the local broadcasting from $S$ to $S'$ is successful if and only if for each node $u \in S$ with transmission power $P_b$,
all the nodes in $S'$ within $u$'s local broadcasting range $r_b$
receives the message from $u$.
We assume we have an algorithm,
\texttt{LocalBroadcast}$(S, S', r_b, P_b)$, which takes sets $S, S'$ of nodes, a
distance $r_b$ and a power level $P_b$ as input, and ensures that the local broadcasting from $S$ to $S'$ is successful.
A small modification of
the local broadcasting scheme of \cite{Goussevskaia+:DialMPOMC08} achieves
this step, which is described below.
\begin{lemma}
\label{lemma:ssma}
Given two sets $S, S'$ of nodes, a local broadcasting range $r_b$ and a power
level $P_b=c'r_b^{\alpha}$, for a constant $c'$,
there is a distributed algorithm \texttt{LocalBroadcast}$(S, S', r_b, P_b)$
that runs in $O(N(S, \gamma r_b) \log |S|)$ time where $N(S, \gamma r_b)$ is the maximum number of nodes in $S$ within a distance $\gamma r_b$ of any node in $S$ and $\gamma$ is a constant, such that:
(i) each node $v\in S'$ receives the message from all nodes in $S$ within
distance $r_b$ \whp, and
(ii) each node $v\in S'$ is able to selectively ignore messages from any node $u\in S'$
which is beyond distance $\gamma r_b$.
\end{lemma}
\begin{proof}[Proof(sketch)]
We describe the protocol as follows.
Following \cite{Goussevskaia+:DialMPOMC08} we use random access, in which
each node in $S$ transmits at Power $P_b$ with prob. $\frac{1}{N(S, \gamma r_b)}$ (known to each node in $S$), and each
node in $S'$ senses the channel to receive messages.
Each time a node $v \in S'$ receives a message,
$v$ checks the total power received;
if that exceeds $\frac{P_b}{(\gamma' r_b)^{\alpha}}+N$, $v$ discards the received message.

Property (i) in Lemma~\ref{lemma:ssma} directly follows from the proofs in \cite{Goussevskaia+:DialMPOMC08}, by partitioning the space into rings and uppering bounding the stochastic interference.
Then Property (ii) follows from the condition we put on accepting a message based on $SP(v)$,
such that any message sent from a node beyond distance $\gamma r_b$ from $v$ will be ignored, where $\gamma$ is a constant.
\end{proof}

\noindent
2.
\textbf{Constant density dominating set}:
We assume an algorithm
\texttt{ConstDominatingSet}$(S, r_c)$ that takes as input a set $S$ of nodes and
a range $r_c \leq r_{max}$, and produces a constant density dominating set $S' \subset S$
corresponding to this range, such that for each node in $S$
there is a node in $S'$ which is with distance $r_c$,
and the density of the nodes in $S'$ is at most a constant.
By density, we mean the number of nodes in $S'$ within a range $r_c$
of any node in $S'$.
The algorithm in \cite{scheideler+:MobiHoc08}
serves this purpose, with the performance summarized below.
\begin{lemma}[\cite{scheideler+:MobiHoc08}]
Given a set $S$ of nodes and a range $r_c \leq r_{max}$, a constant density dominating set
for $S$ can be constructed in time $O(\log{n})$ \whp
under the SINR model.
\end{lemma}

\subsection{Nearest Neighbor Tree Scheme}
The algorithmic paradigm underlying our algorithm is the Nearest Neighbor Tree (NNT) scheme
\cite{khan_tpds, khan_tcs, khan_distcomp}.
The NNT scheme can be used to construct a spanning tree,
called the {\em Nearest Neighbor Tree (NNT)}, efficiently in a
distributed fashion. The  cost of the NNT can be shown  to be
within an $O(\log n)$ factor of the cost of the MST. The scheme
used to construct an NNT (henceforth called  {\em NNT scheme})
consists of the following two steps:
\begin{compactenum} [(1)]
\item each node first chooses a unique \emph{rank} from a
totally-ordered set; a ranking of the nodes corresponds to a
permutation of the nodes;
\item each node (except the one with the
highest rank) connects (via the {\em shortest path}) to the
\emph{nearest} node of higher rank.
\end{compactenum}

It can be shown that the NNT scheme
constructs a spanning subgraph in any weighted graph whose cost is
at most $O(\log n)$ times that of the MST, irrespective of how the
ranks are selected (as long as they are distinct)
\cite{khan_tpds,khan_tcs}. Note that some cycles can be introduced in
step 2, and hence to get a spanning tree we need to remove some
edges to break the cycles.
The main advantage of the NNT scheme is that each node,
individually, has the task of finding its nearest node of higher
rank to connect to, and hence no explicit coordination is needed
among the nodes.

However, despite the simplicity of the NNT scheme,
it is non-trivial to efficiently implement the scheme in an arbitrary weighted graph.
The work of \cite{khan_distcomp} showed how to efficiently implement the scheme in an arbitrary
weighted graph in the CONGEST (wired) model. This implementation was shown to have
a running time of $\tilde{O}(D(G) + L(G,w))$
where $L(G,w)$ is a parameter called the {\em local shortest path
diameter} and $D(G)$ is the (unweighted) diameter of the graph.  This distributed implementation cannot be directly used here,
due to the (additional) complication of having to obey SINR constraints when each node   tries to search
for the nearest node of higher rank to connect to.
Second, the local shortest path diameter
can be significantly larger than the diameter of the underlying graph. In particular, it can
be as large as $n$, in which case, the above
implementation does not give the time bound that we would like to show in this paper, i.e., close
to the diameter of the underlying  graph. Note that this is essentially the best possible,
since computing a spanning tree on an arbitrary graph takes diameter time.
To show this stronger bound, we exploit the
fact that the underlying graph has a geometric structure. The main idea
is to choose ranks in a particular way (in step (1) of the NNT scheme) that guarantees the
each node can find its nearest node of higher rank fast. This idea was first implemented in  the UDG-NNT (Unit disk graph-NNT) algorithm \cite{khan_tpds}. However, again this implementation does not  directly work
in the SINR model. We show that one has to implement both step (1) and step (2) of the NNT scheme
in a incremental and staged fashion, so that SINR constraints are obeyed and still many
nodes can progress simultaneously. We show that NNT technique results in spatial separation of ``active" nodes (i.e., nodes that need to communicate) in each step, and hence is amenable to obey SINR constraints.  The   details are in the next Section.

%% file: bounded-power.tex
\section{Approximating the MST in the SINR model}
\label{sec:algo}

We now discuss the algorithm \texttt{MST-SINR}
for finding a low cost spanning tree.
We start with the main intuition for the algorithm.
The basic idea for satisfying SINR constraints
at each step is to ensure that the senders are ``spatially separated''.
In many distributed MST algorithms, e.g., the algorithms of \cite{DistMst:Gallager,garay-sublinear}, there is no
guarantee that senders are well separated. However, the
{\em Nearest Neighbor Tree (NNT)} scheme \cite{khan_tpds, khan_tcs, khan_distcomp}
has a nice property that nodes gradually increase their range, and the nodes
which are active in any round are sufficiently well separated. However, this
approach has only been studied in CONGEST model in \cite{khan_tpds, khan_tcs, khan_distcomp},
and we adapt it to the SINR model.
For illustration, first imagine that
the maximum power $P_{max}$ is ``unbounded'' so that the $d_{max}\leq r_{max}$.
The algorithm in this special case involves the following steps.
\begin{compactenum} [(1)]
\item
We have $\log{r_{max}} \leq \mu$ phases, ranging from $i=1,\ldots,\log{r_{max}}$.
\item
In the $i$th phase, a subset $S_i$ of nodes participate, and the edges
chosen so far form a forest rooted at nodes in $S_i$. The nodes in $S_i$
transmit at power level of $c \cdot d_i^{\alpha}$ for a constant $c$, where $d_i= 2^i$.
\item
In the $i$th phase, each node $v\in S_i$ runs the NNT scheme: $v$ connects to a ``close-by''
node in $S_i$ within distance $c'\cdot d_i$ of higher rank, if one exists, for
a constant $c'$. The nodes which
are not able to connect continue into phase $i+1$.
\end{compactenum}

We need to prove that the above steps can be implemented in the SINR model, and the
resulting implicit tree is a low cost tree, relative to the MST. One complication,
in contrast to the original NNT scheme, is that there is no way to ensure that a node
$v\in S_i$ connects to the ``closest'' node in $S_i$ of higher rank, within distance
$d_i$. Because of the SINR model, it is possible that a transmission from some
far away node could be received by $v$. We show that the probability of this event is
very low, so that with high probability, we get a tree that is ``close to'' the NNT,
leading to the logarithmic approximation bound for the cost of the tree.
We describe the above algorithm as subroutine
\texttt{NNT-SINR-BP} in Section \ref{sec:powercon}, with the modification that
$P_{max}$ is bounded, and instead of $\mu$ phases, we only have $\log{r_{max}}$
phases. The result of \texttt{NNT-SINR-BP} is a forest with the highest rank nodes
forming the roots.

In the general case, $P_{max}$ is bounded, so that $d_{max}>r_{max}$, and all we have
is that the graph induced by range $r_{max}$ is connected\footnote{Recall that in Section
\ref{sec:model}, we in fact assume something stronger: the graph induced by
a range of $r_{max}/c$ is connected for a constant $c$.}.
If we run the above
bottom-up algorithm with some maximum range $r_1$,
we would get a forest in which the roots are at least $c\cdot r_1$ apart, where
$c$ is a constant. In order to connect up the roots, we first use
a ``top-down'' approach where we choose a set $Dom$ of nodes so that (i) each node
$v\not\in Dom$ is within distance $r_1$ of some node in $Dom$, and (ii) for each node
$u\in Dom$, there are a ``small'' number of nodes within distance $r_1$. The idea is
that if we can construct a spanning tree on $Dom$ (in the SINR model), and can
ensure that each node in $V \setminus Dom$ has a higher rank neighbor in $Dom$, within distance
$r_1$, then the above bottom-up phase restricted to $V \setminus Dom$ will allow the
trees to be combined. Such a set $Dom$ is precisely a ``constant density dominating set'',
as defined and computed efficiently in \cite{scheideler+:MobiHoc08}.
The spanning tree on $Dom$ is constructed by adapting
the UDG-NNT algorithm of \cite{khan_tpds}, which starts with a node $s$ and spreads the
ranks from it. This adaptation is discussed below as subroutine \texttt{NNT-SINR-CD},
which also works in the SINR model because the nodes in $Dom$ are spatially well separated.

We first describe the two subroutines, and then discuss the main algorithm.

\subsection{Constructing forests with power constraints}
\label{sec:powercon}

\begin{algorithm}[htbp]
\caption{\texttt{NNT-SINR-BP}$(S, P_{max}, rank(\cdot))$\label{alg:nnt-sinr-bp}}
\SetNoFillComment
\SetKwComment{Comment}{$\#$ }{}
\SetKwComment{tcc}{$\#$ }{}
\SetKwInOut{InputP}{Input}
\SetKwInOut{OutputP}{Output}

\InputP{node set $S$, max power $P_{max}$, and value $rank(v)$ for each $v\in S$}
\OutputP{set $F$ of edges}

Each node $v \in S$ does the following:

$r_{max} \gets (P_{max}/c)^{1/\alpha}$, $i_{max} \gets \lfloor \log_2 r_{max} \rfloor$\;

initially $S_i \gets \emptyset, \forall i>1$;
\fontsize{8.5}{10.5pt}\selectfont
\Comment{$S_i$ denotes the set of active nodes in phase $i$}\normalsize
$S_1 \gets S$\;

\For(\fontsize{8.5}{10.5pt}\selectfont\tcc*[h]{in each phase}\normalsize){$i=1$ to $i_{max}$}{
            \fontsize{8.5}{10.5pt}\selectfont
    \Comment{$v$ broadcasts its rank to all nodes in $S_i$ within range $d_i$}\normalsize
    $d_i \gets 2^i$, $P_i \gets c' d_i^{\alpha}$; \fontsize{8.5}{10.5pt}\selectfont\Comment{$c'$ is a constant}\normalsize
    \If(\fontsize{8.5}{10.5pt}\selectfont\tcc*[h]{broadcast rank value}\normalsize){$v \in S_i$} {
        $v$ broadcasts its rank by participating in \texttt{LocalBroadcast}$(S_i, S_i, d_i, P_i)$\;
    }
    $S_v \gets$ set of nodes with ranks received by $v$\;
    $v' \gets$ the node with highest rank in $S_v$\;
    \uIf(\fontsize{8.5}{10.5pt}\selectfont\tcc*[h]{$v$ has highest rank locally: add $v$ to the active set $S_{i+1}$ for next phase}\normalsize){$v = v'$}{
      add $v$ to $S_{i+1}$\; 
    }
    \lElse{
      $par(v) \gets v'$; \fontsize{8.5}{10.5pt}\selectfont\Comment{$v$ is done}\normalsize
    }
}

$F \gets \{(v, par(v)): par(v) \neq \emptyset, \forall v \in S\}$\;
\end{algorithm}

For subroutine \texttt{NNT-SINR-BP}, we are given a set $S$,
a maximum power level $P_{max}$ (corresponding to a range $r_{max}$),
and a rank function $rank(\cdot)$, which
assigns unique ranks for all nodes in $S$. The goal is to construct
a forest, in which each node connects to a parent within range $r_{max}$.
As discussed earlier, \texttt{NNT-SINR-BP} uses the NNT approach of
\cite{khan_tpds, khan_tcs, khan_distcomp}, in which each node connects
to the nearest node of higher rank, which leads to a forest. However,
in order to be feasible in the SINR model, we need to do this in a careful
and staged manner, so that the set of transmitting nodes in each round form
a constant density set.

\begin{lemma}
\label{lemma:constant-ball}
At the beginning of each phase $i$, for each node $v \in S_i$,
there are at most a constant number of active nodes in the ball $B(v, d_i)$, \ie,
$\big|B(v, d_i) \bigcap S_i\big| = O(1)$. Further, for each node $v \in S_i$,there are no other active nodes in the ball $B(v, d_i / 2)$, \ie,
$B(v, d_i/2) \bigcap S_i = \{v\}$.
\end{lemma}

\begin{proof}[Proof by induction]
Recall that $S_{i}$ is the set of active nodes at phase $i$.
For phase 1, $S_1 = S$ and $d_0 = 1$; obviously there can be at most 16 nodes in the ball.
Assume that for phase $i$, the statement is true.
Then, at the beginning of phase $i+1$, if $v$ is active, according to the algorithm,
all nodes in the range of $d_i$ of $v$ should all have a lower rank and thus are all inactive.
Therefore, the distance between any two nodes in $S_{i+1}$ is at least $d_i$.
By packing property, there can be at most 16 nodes in the ball $B(v, d_{i+1})$.
Therefore, the statement is true for every phase.
\end{proof}

\begin{lemma}
\label{lemma:connected-all}
In the end of the algorithm, the set $F$ forms a forest. Further,
the set $S_{i_{max}}$ forms the set of roots of $F$, and for any two
nodes $u,v\in S_{i_{max}}$, we have $d(u,v)\geq r_{max} / 2$.
\end{lemma}
\begin{proof}
Because of the NNT property, in which each node only connects to a higher
rank parent, there are no cycles. Next, by design, the nodes in $S_{i_{max}}$
do not connect to any other node, and form the roots of the forest $F$.
By Lemma \ref{lemma:constant-ball}, it follows that for any
$u,v\in S_{i_{max}}$, we have $d(u,v)\geq r_{max} / 2$.
\end{proof}


\begin{lemma}
\label{lemma:nnt-sinr-bp}
The cost of the forest $F$ is at most $O(\mu)$ times that of
of the minimum spanning forest of the nodes in $F$.
Further,
{\em \texttt{NNT-SINR-BP}} is feasible in the SINR model, with running time of
$O(\mu \log n)$, \whp
\end{lemma}

\begin{proof}
For $i<i_{max}$, let $F_i=\{(v,w) \in F:v\in S_i, w=par(v)\in S_{i+1}\}$
denote the set of edges in $F$ between nodes in sets $S_i$ and $S_{i+1}$.
Recall that $S_i$ denotes the set of nodes which remain active in the beginning of phase $i$.
Let $MST(S_i)$ denote a minimum spanning tree of only the nodes in $S_i$.
W.l.o.g., we assume $|S_i|>1$.
For any pair of nodes in $S_i$, the mutual distance is at least $d_i / 2$
(from Lemma \ref{lemma:constant-ball}), it follows that
$cost(MST(S_i)) \geq (|S_i|-1) d_i / 2 \geq |S_i| d_i / 4$.
Next, we have $cost(F_i) \leq d_i |F_i| \leq d_i |S_i|$, implying
$cost(F_i) \leq 4 cost(MST(S_i))$.
Since there are $O(\log r_{max})$ phases and $cost(MST(S_i)) \leq cost(MST)$,
we have
$cost(F) = \sum_{i} cost(F_i) \leq 4 (\log r_{max}) \sum_i cost(MST(S_i))$.
We have assumed $r_{max} \leq d_{max}$;
therefore, $\log r_{max} \leq \mu$, implying the statement of approximation ratio.
By Lemmas~\ref{lemma:ssma} and \ref{lemma:constant-ball},
each run of the local broadcast primitive finishes in $O(\log n)$ time, so does each phase of the algorithm.
The statement of running time follows.
\end{proof}

\subsection{Spanning trees for instances with constant density}
\label{sec:cd}

\begin{algorithm}[htbp]
\caption{\texttt{NNT-SINR-CD}$(S, s, r)$ \label{alg:nnt-sinr-cd}}
\SetNoFillComment
\SetKwComment{Comment}{$\#$ }{}
\SetKwComment{tcc}{$\#$ }{}
\SetKwInOut{InputP}{Input}
\SetKwInOut{OutputP}{Output}

\InputP{constant density node set $S$, sink $s$, range $r$}
\OutputP{spanning tree $T$ on $S$}

Each node $v \in S$ does the following:

\fontsize{8.5}{10.5pt}\selectfont
\Comment{in phase 0:}\normalsize
\uIf(\fontsize{8.5}{10.5pt}\selectfont\tcc*[h]{$s$ broadcasts to neighbors}\normalsize){$v = s$}{
    $v$ generates a large $rank(v)$, and broadcasts it by \texttt{LocalBroadcast}$(\{s\}, S, r, c r^{\alpha})$\;
}
\lElse{
    $v$ listens;
}

\fontsize{8.5}{10.5pt}\selectfont
\Comment{in each phase $i$:}\normalsize
\For{$i=1$ to $2 D(G_r(S))$}{
    $S_i \gets \emptyset$\;  \label{algoline:nntcd-start}
    \If{$v$ receives a message $m=(id', rank')$}{
        \If{$m$ is the first message $v$ ever receives}{
            $v$ generates $rank(v)$ randomly such that $rank(v)<rank'$\;
            $S_i \gets S_i \cup \{v\}$\;
            $v$ participates in \texttt{LocalBroadcast}$(S_i, S, r, c r^{\alpha})$\; \label{algoline:nntcd-broadcast}
        }
        add $id'$ to list $L(v)$, \textbf{if} $rank' > rank(v)$\;
    } \label{algoline:nntcd-end}
}

$par(v) \gets$ one of the nodes in $L(v)$\;

\Return the set $\{(v, par(v)): v\neq s\}$ of edges
\end{algorithm}


We now consider the problem of constructing a low cost spanning tree for
a constant density instance.
Subroutine \texttt{NNT-SINR-CD} takes such a set $S$, a sink $s$ and
a range $r$ as input. As defined in Section~\ref{sec:model},
$G_r(S)$ denotes the $S$-induced graph based on range $r$.
We adapt the UDG-NNT algorithm of \cite{khan_tpds} for producing a spanning tree.  Let
$D(G_r(S))$ denote the diameter of $G_r(S)$; we assume (an estimate of) $D(G_r(S))$ is known to
all the nodes. We note that in Line~\ref{algoline:nntcd-broadcast} of the algorithm, the local broadcasts
are run simultaneously for all nodes
$v$, which get the rank message for the first time. Further, each iteration of the
\textbf{for} loop in Lines~\ref{algoline:nntcd-start}-\ref{algoline:nntcd-end} involves $O(\log{n})$ time steps,
in which the local broadcast is run; nodes $v$ which do not have to send their
ranks in some iteration remain silent during these steps.

\begin{lemma}
\label{lemma:nnt-sinr-cd}
If $G_r(S)$ is connected, and $S$ is a constant density set with range $r$,
Algorithm {\em \texttt{NNT-SINR-CD}} produces under the SINR model a spanning tree $T_1$ on $S$ with
$cost(T_1)=O(OPT(G_r(S)))$ in time $O(D(G_r(S)) \log{n})$ with high probability,
where $OPT(G_r(S)) = cost(MST(G_r(S)))$.
\end{lemma}
\begin{proof}
Our proof mimics the proof of algorithm UDG-NNT from \cite{khan_tpds}.
Since $G_r(S)$ is connected,
each node $v\neq s$ is able to run
line 12, and the tree $T_1$ is constructed. The simultaneous calls
to \texttt{LocalBroadcast} all take time $O(\log{n})$ in each round (with
high probability), so that the overall running time is
$O(D(G_r(S)) \log{n})$, \whp. Since the \texttt{LocalBroadcast}
algorithm is feasible in the SINR model, \texttt{NNT-SINR-CD} is also feasible.

The constant density property of $S$ implies that any ball $B(v, r')$ with $r'\leq r$
has $O(1)$ nodes in $S$. Therefore, $cost(MST(G_r(S))) =\Omega(|S| r)$. Next, the
local broadcasts (Lemma \ref{lemma:ssma}) ensure that each
node $v$ receives messages from nodes within
distance $c' r$ for a constant $c'$. Therefore, $cost(T_1)\leq |S| c' r$,
and the lemma follows.
\end{proof}

\subsection{Putting everything together: Algorithm \texttt{MST-SINR}}
Our high level idea is to start
with a constant density dominating set $Dom$, and run subroutine \texttt{NNT-SINR-CD}
to construct a tree $T_1$ on $Dom$. We then choose and disseminate ranks suitably, and run
\texttt{NNT-SINR-BP} to form a forest, which get connected to the nodes in $Dom$,
and together form a spanning tree.

\begin{algorithm}[htbp]
\caption{MST-SINR$(V)$ \label{alg:nnt-sinr-bp}}
\SetNoFillComment
\SetKwComment{Comment}{$\#$ }{}
\SetKwComment{tcc}{$\#$ }{}
\SetKwInOut{InputP}{Input}
\SetKwInOut{OutputP}{Output}

\InputP{node set $V$}
\OutputP{spanning tree $T$}

\fontsize{8.5}{10.5pt}\selectfont
\Comment{use the algorithm of \cite{scheideler+:MobiHoc08} to construct a
constant density dominating set $Dom$ w.r.t. range $r_{max}/c$,
for a constant $c$}\normalsize
$Dom \gets \texttt{ConstDominatingSet}(V, r_{max}/c)$ for a constant $c$\;
\label{algoline:const-dom}

\fontsize{8.5}{10.5pt}\selectfont
\Comment{we adapt the UDG-NNT algorithm of \cite{khan_tpds} for constructing a
spanning tree on $Dom$}\normalsize
$s \gets$ a node in $Dom$; \fontsize{8.5}{10.5pt}\selectfont\Comment{sink node}\normalsize
$T_1 \gets$ \texttt{NNT-SINR-CD}$(Dom, s, 2r_{max}/c)$\;
\For{each $v\in Dom$}{
$rank(v) \gets$ the rank chosen in this process\;
}

\fontsize{8.5}{10.5pt}\selectfont
\Comment{each $v\in Dom$ broadcasts its rank to nodes within its range}\normalsize
run \texttt{LocalBroadcast}$(Dom, V \setminus Dom, r_{max}/c)$ such that all $v\in Dom$
to broadcast $rank(v)$\;

\For{each $v\in V \setminus Dom$}{
    $b(v) \gets$ largest rank received by $v$\;
    $q(v) \gets$ id of the node that sent $v$ rank $b(v)$\;

    $v$ chooses $rank(v)$ uniformly at random such that $rank(v)<b(v)$\;

    $dom(v) \gets q(v)$; \label{algoline:dom-node}
    \fontsize{8.5}{10.5pt}\selectfont\Comment{to connect to tree roots of forest $F_1$ constructed below}\normalsize
}

\fontsize{8.5}{10.5pt}\selectfont
\Comment{use the ranks chosen above to build a forest with tree roots in $Dom$ to connect
up with the tree $T_1$}\normalsize
$F_1 \gets$ \texttt{NNT-SINR-BP}$(V \setminus Dom, P_{max}, rank(\cdot))$\; \label{algoline:nnt-sinr-bp}

\For{each $v\in V \setminus Dom$}{
    \If(\fontsize{8.5}{10.5pt}\selectfont\tcc*[h]{$v$ is a tree root in forest $F_1$}\normalsize){$par(v) = \emptyset$}{
        $par(v) \gets dom(v)$; \label{algoline:connect-to-t1}
        \fontsize{8.5}{10.5pt}\selectfont\Comment{connect to $T_1$}\normalsize
    }
}
\Return set $\{(v, par(v)): par(v) \neq \emptyset\}$ of edges.
\end{algorithm}

Our analysis, at a high level, involves the following steps: we first show that
$T_1$ is a spanning tree on $Dom$ with range $2r_{max}/c$, and has low cost.
Next, we show that
Algorithm \texttt{NNT-SINR-BP} with the range of $r_{max}$ results in
a forest on $V \setminus Dom$, each of whose components gets connected to some node in $Dom$,
because of the way the ranks are chosen. Finally, we show that the
combined tree produced in this manner has low cost.

\begin{lemma}
\label{lemma:Dconn}
The graph $G_{2r_{max}/c}(Dom)$ induced by $Dom$ w.r.t. a range of $2r_{max}/c$
is connected, and has a diameter of at most $2D$.
\end{lemma}
\begin{proof}
Suppose $G_{2r_{max}/c}(Dom)$ is not connected.
Then there exists a partition $Dom=Dom_1\cup Dom_2$, which induce disconnected
components, and $\min_{(u,v)\in Dom_1\times Dom_2}\{d(u,v)\}>2r_{max}/c$,
since the dominating set $Dom$ constructed uses range $r_{max}/c$.
That implies that $G$ is connected with $r_{max}/c$,
which is a contradiction.
Therefore, $G_{2r_{max}/c}(Dom)$ is connected.
\end{proof}

\begin{lemma}
\label{lemma:nnt-sinr-bp1}
Algorithm {\em\texttt{MST-SINR}} produces a spanning tree of cost $O(\mu)$
times the optimal in time $O(D\log{n}+\mu\log{n})$, with high probability,
in the SINR model.
\end{lemma}
\begin{proof}
From Lemma \ref{lemma:nnt-sinr-cd}, the call to \texttt{NNT-SINR-CD} produces
a spanning tree $T_1$ on set $Dom$. Next, consider the call to \texttt{NNT-SINR-BP}
in Line~\ref{algoline:nnt-sinr-bp} of the algorithm.  Let $V'=\{v\in V \setminus Dom: par(v)=\phi\}$ be the
set of nodes for whom the parent is not defined during this call.
Each nodes in $V'$ corresponds to a tree root of the forest constructed in this call,
\ie, the edges $F_1=\{(v, par(v)): v\in V \setminus Dom, par(v)\neq\phi\}$ (constructed during this call)
form a forest, rooted at nodes in $V'$.
In Line~\ref{algoline:dom-node}, each node in $V \setminus Dom$ registers its dominating node in $T_1$;
in Line~\ref{algoline:connect-to-t1}, each node $v \in V'$ uses this information
to connect to $T_1$.
Therefore, $T_1\cup F_1\cup\{(v,par(v)): \forall v\in V'\}$ is a spanning tree.

From Lemma \ref{lemma:nnt-sinr-cd}, we have $cost(T_1)=O(OPT)$.
Next, from Lemma \ref{lemma:nnt-sinr-bp}, it follows that $cost(F_1)=O(\mu \cdot OPT)$.
Finally, for each $v\in V'$, $d(v, par(v))\leq c'r_{max}/c$. Therefore,
$\sum_{v\in V'} d(v, par(v))=O(OPT)$, which implies the cost of the tree
produced by \texttt{NNT-SINR-BP} is $O(\mu \cdot OPT)$.

Finally, we analyze the running time. The call to
\texttt{ConstDominatingSet} in Line~\ref{algoline:const-dom} takes $O(\log{n})$ time. From
Lemma \ref{lemma:nnt-sinr-cd}, the call to \texttt{NNT-SINR-CD} takes
time $O(D\log{n})$, with high probability. Next, the local broadcast
takes time $O(\log{n})$, and the call to \texttt{NNT-SINR-BP} takes time
$O(\mu\log{n})$. Putting all of these together, the total running time is
$O(D\log{n}+\mu\log{n})$, with high probability. From \cite{scheideler+:MobiHoc08},
and Lemmas \ref{lemma:nnt-sinr-cd}, \ref{lemma:nnt-sinr-bp} and
\ref{lemma:ssma}, all the computations of the algorithm are feasible in the SINR model.
\end{proof}

\subsection{Scheduling complexity of the spanning tree}
\label{sec:sched-complexity}

The \emph{scheduling complexity} (as defined in \cite{Moscibroda+:MobiHoc06}) of a set of communication requests is the minimal number $t$ of time slots during which each link request has made at least one successful transmission.
Due to the constant factor approximation algorithm by Kesselheim~\cite{Kesselheim:SODA11} for finding a maximum feasible set of links,
any spanning tree has a scheduling complexity of $O(\log n)$ under the SINR model. 
Here, we use the term \emph{distributed scheduling complexity} to describe the scheduling complexity for 
a set of links that needs to make transmission decisions in a distributed fashion. 
The spanning tree $T$ produced by our algorithm has the following property.

\begin{lemma}
\label{lemma:scheduling-complexity}
The spanning tree $T$ produced by Algorithm {\em \texttt{MST-SINR}}
has a distributed scheduling complexity of $O(\mu \log n)$, regardless of the direction of the transmission requests on the edges.
\end{lemma}
\begin{proof}[Proof by construction]
Let the edges in $T$ be oriented in arbitrary directions, so that we obtain a set $L$ of links,
which represent a set of transmission requests.

We construct a feasible schedule which can be implemented in a distributed fashion.
We partition the links into length classes $1,2,\ldots,\log r_{max}$,
such that all links in length class $i$ have a length between $2^{i-1}$ and $2^i$.
In each phase $i \in [1,\log r_{max} / c]$, links in length class $i$ transmit
with probability $1/K$ in each time step,
where $K$ is an upper-bound of the number of edges/links in length class $i$ with an end node within range $2^{i+1}$ of each node.
With the same proof approach as that in \cite{Goussevskaia+:DialMPOMC08}, we can prove that
in $O(K \log n)$ time steps, all the links in length class $i$ have made successful transmissions with high probability.
This way, in $\log r_{max} / c]$ phases, all the links have made successful transmissions, \whp

Now, we argue that $K$ is $O(1)$ in each phase, and thus the total schedule length is $O(\mu \log n)$.
Due to the construction of $T$ in Algorithm \texttt{MST-SINR}, since $Dom$ is already a constant density node set,
we only need to analyze the edges formed during the call to
\texttt{NNT-SINR-BP}$(V \setminus Dom, P_{max}, rank(\cdot))$ in Line~\ref{algoline:nnt-sinr-bp} of Algorithm \texttt{MST-SINR}.
In each phase of the \texttt{NNT-SINR-BP}, the lengths of edges increases upon the previous phase.
In one phase during \texttt{NNT-SINR-BP}, the newly formed edges corresponds to a length class,
and the nodes involved in that phase form a constant density node set, according to
Lemma~\ref{lemma:constant-ball}.
That implies that the number of edges in length class $i$ with an end node
within range $2^{i+1}$ of each node is $O(1)$.

In terms of the distributed implementation of scheduling, each edge may remember the phase number
in \texttt{NNT-SINR-BP} when they are connected up by a new edge,
and participate in the phase of the same phase number of our scheduling construction.
\end{proof}

%% file: conclusion.tex
\section{Conclusion}
\label{sec:conclusion}

In this paper, we describe the first distributed  algorithm in the SINR
model for approximate minimum spanning tree construction. This is the first
such result for solving a ``global'' problem in the emerging SINR
based distributed computing model--- this is in contrast to ``local'' problems, such
as independent sets and scheduling, for which distributed algorithms are
known in this model. Our algorithm produces a logarithmic approximation
to the MST, and takes time $O(D\log{n}+\mu\log{n})$, featuring a
distributed scheduling complexity of $O(\mu \log n)$.
Our main technical
contribution is the use of \emph{nearest neighbor trees}, which naturally
ensure spatial separation at each step, thereby allowing SINR constraints
to be satisfied.